\documentclass[envcountsect, envcountsame, 11pt]{llncs}

\usepackage{llncsdoc}

\usepackage{amsmath,amssymb}

\usepackage[lined,boxed]{algorithm2e}
\usepackage{graphicx, xcolor}

\newcommand{\cl}{\textnormal{codelength}}
\newcommand{\sat}{\textnormal{sat}}

\newcommand{\pweak}[1]{\ensuremath{P^{\rm weak}_{#1}}}
\newcommand{\pstrong}[1]{\ensuremath{P^{\rm strong}_{#1}}}
\newcommand{\prweak}[1]{\ensuremath{\vdash^{\rm weak}_{#1}}}
\newcommand{\prstrong}[1]{\ensuremath{\vdash^{\rm strong}_{#1}}}

\newcommand{\scalar}[2]{\ensuremath{\langle #1, #2\rangle}}
\newcommand{\ignore}[1]{}

\newcommand{\E}{\mathop{\mathbb{E\/}}}
\newcommand{\ftwo}{\ensuremath{\mathbb{F}_2}}

\newtheorem{observation}[theorem]{Observation}

\newcommand{\rank}{\ensuremath{\textnormal{rank}}}

\newcommand{\floor}[1]{\left\lfloor #1 \right\rfloor}
\newcommand{\ceil}[1]{\left\lceil #1 \right\rceil}
\renewcommand{\vec}[1]{\mathbf{#1}}
\newcommand{\pfrac}[2]{\left(\frac{#1}{#2}\right)}
\newcommand{\ppsz}{\textup{\texttt{ppsz}}}

\newcommand{\encode}{\textup{\texttt{encode}}}

\renewcommand{\r}{\mathbf{r}}
\newcommand{\x}{\mathbf{x}}
\newcommand{\y}{\mathbf{y}}
\renewcommand{\u}{\mathbf{u}}
\newcommand{\boldv}{\mathbf{v}}
\newcommand{\w}{\mathbf{w}}
\newcommand{\z}{\mathbf{z}}
\renewcommand{\a}{\mathbf{a}}

\newcommand{\N}{\ensuremath{\mathbb{N}}}

\makeatletter
\newenvironment{oneshot}[1]{\@begintheorem{#1}{\unskip}}{\@endtheorem}
\makeatother

\definecolor{violet}{rgb}{1.0,0.0,1.0}
\definecolor{darkgreen}{rgb}{0,0.3,0}

\newif{\ifhidecomments}

\hidecommentstrue

\ifhidecomments

 \newcommand{\dominik}[1]{}
 \newcommand{\navid}[1]{}
 \newcommand{\pavel}[1]{}
 
 \else

 \newcommand{\dominik}[1]{{\color{blue}\textless dominik\textgreater\ #1\ \textless/dominik\textgreater}}
 \newcommand{\navid}[1]{{\color{violet}\textless navid\textgreater\ #1\ \textless/navid\textgreater}}
 \newcommand{\pavel}[1]{{\color{darkgreen}\textless pavel\textgreater\ #1\ \textless/pavel\textgreater}}

\fi

\newcommand{\nullv}{\mathbf{0}}
\renewcommand{\b}{\mathbf{b}}
\newcommand{\e}{\mathbf{e}}
\newcommand{\F}{\mathbb{F}}

\newcommand{\ttconstant}{60}
\newcommand{\llconstant}{5}

\title{Tighter Hard Instances for PPSZ}

\author{Pavel Pudl{\'a}k\thanks{The author is supported by the grant P202/12/G061 of GA\v{C}R.}\inst{1}  \and Dominik Scheder\thanks{Dominik Scheder gratefully acknowledges
support by the National Natural Science Foundation of China under grant 61502300.}\inst{2}
\and Navid Talebanfard\thanks{Part of the work was done while Navid Talebanfard was with Tokyo Institute of Technology and visiting Saint Petersburg State University during the special semester in complexity.}\inst{3}}

\institute{The Czech Academy of Sciences \and Shanghai Jiaotong University \and
Saarland University and the Cluster of Excellence, MMCI}

\begin{document}

\maketitle

\begin{abstract}
We construct uniquely satisfiable $k$-CNF formulas that are hard for the algorithm PPSZ. 
Firstly, we construct graph-instances on which ``weak PPSZ'' has savings of at most $(2 + \epsilon) / k$;
the {\em saving} of an algorithm on an input formula with $n$ variables is the largest
$\gamma$ such that the algorithm succeeds (i.e. finds a satisfying assignment) with probability
at least $2^{ - (1 - \gamma) n}$. Since PPSZ (both weak and strong) is known to have savings of at least
$\frac{\pi^2 + o(1)}{6k}$, this is optimal up to the constant factor. 
In particular, for $k=3$, our upper bound is $2^{0.333\dots n}$, which is fairly close to the lower bound $2^{0.386\dots n}$ of Hertli [SIAM J. Comput.'14]. We also construct instances based on linear systems over $\F_2$ for which {\em strong} PPSZ has savings of at most $O\pfrac{\log(k)}{k}$. This is only a $\log(k)$ factor away from 
the optimal bound. Our constructions improve previous savings upper bound of $O\pfrac{\log^2(k)}{k}$ due to Chen et al. [SODA'13].
\end{abstract}

\dominik{SODA submission format: The submission, excluding title page, diagrams, bibliography and appendix, must not exceed 10 pages.  (Authors should feel free to send submissions that are significantly shorter than 10 pages.)\\
Deadlines:
\begin{itemize}
\item July 6, 2016, 4:59 PM EDT - Deadline - Short Abstract Submission and Paper 
Registration Deadline
\item July 13, 2016, 4:59 PM EDT - Deadline - Full Paper Submission
\end{itemize}
} 

\navid{Navid can also make comments by typing something like
\textbackslash navid\{Navid can also make comments by typing something like
\textbackslash navid\{Navid can\dots \} \dots \}.}

\pavel{And so does Pavel.}

\newpage
\setcounter{page}{1}

\section{Introduction}

The $k$-SAT problem is one of the most fundamental NP-complete problems: given a $k$-CNF formula decide if there is an assignment to the variables that satisfies all the clauses. While a simple exhaustive search algorithm solves the problem, attempting to beat this trivial approach remains an active direction (see e.g. \cite{PPZ99,S99,D02,PPSZ05}). Formalizing the true hardness of $k$-SAT, Impagliazzo and Paturi \cite{IP01} presented two hypotheses: {\it Exponential Time Hypothesis (ETH)} which rules out any $2^{o(n)}$ time algorithm for $k$-SAT where $n$ is the number of variables and {\it Strong Exponential Time Hypothesis (Strong ETH)} which says that for any $\epsilon > 0$ there exists $k > 0$ such that $k$-SAT cannot be solved in time $2^{(1 - \epsilon)n}$. Both ETH and Strong ETH have successfully been used to explain the hardness of many other problems; under ETH one can prove tight lower bounds for many fixed parameter tractable problems (see \cite{LMS11}), and under Strong ETH several lower bounds for polynomial time solvable problems are proved (see e.g. \cite{BI15}). However the validity of both these hypotheses remains a matter of mystery and in particular regarding Strong ETH no consensus seems to be within reach any time soon.

In this paper we focus on Strong ETH and the problem of constructing hard instances for known classes of algorithms for $k$-SAT. 
Paturi, Pudl{\'a}k, Saks and Zane \cite{PPSZ05} presented the currently best known randomized algorithm for $k$-SAT. The algorithm roughly does the following: pick uniformly at random a variable $x$ from the input formula. Try to infer the value of $x$ using some sound heuristic. If this check fails, pick a random value for $x$. Set $x$ to be this value and repeat. A sound heuristic is an algorithm $P$ that receives a formula $F$ and a variable $x$ such that $P(F, x) = 0$ implies $F \models (x = 0)$ and $P(F, x) = 1$ implies $F \models (x = 1)$. We will consider two heuristics, $\pweak{w}$ amd $\pstrong{w}$, where $\pweak{p}$ checks if the value of $x$ can be derived from any set of $w$ clauses of $F$, and $\pstrong{w}$ checks if the value of $x$ can be derived by a width-$w$ resolution derivation from $F$. Note that if $F$ is $O(1)$-CNF then both $\pweak{w}$ and $\pstrong{w}$ run in subexponential time as long as $w = o(\frac{n}{\log n})$. The first result showing that even simple sound heuristics can yield non-trivial savings over exhaustive search was proved by Paturi, Pudl{\'a}k and Zane.

\begin{theorem}[\cite{PPZ99}]
\label{thm:ppz}
Let $F$ be a $k$-CNF formula on $n$ variable. Then $$\Pr[\ppsz(F, \pweak{1}) \in \sat(F)] \ge 2^{-(1 - \frac{1}{k})n}.$$
\end{theorem}

Naturally one can ask if stronger heuristics can improve the success probability. It was indeed shown in the following theorem that using  $\omega(1)$-width resolution yields improvements.

\begin{theorem}[\cite{PPSZ05}]
\label{thm:ppsz}
Let $F$ be a $k$-CNF formula on $n$ variables. Then $$\Pr[\ppsz(F, \pstrong{\omega(1)}) \in \sat(F)] \ge 2^{-(1 - \frac{\pi^2}{6k}-o(1))n}.$$ 
\end{theorem}

Later Hertli \cite{H14} showed among other things that even $\pweak{\omega(1)}$ yields the same improvement over the trivial $\pweak{1}$.

\begin{theorem}[\cite{H14}]
\label{thm:hertli}
Let $F$ be a $k$-CNF formula on $n$ variables. Then $$\Pr[\ppsz(F, \pweak{\omega(1)}) \in \sat(F)] \ge 2^{-(1 - \frac{\pi^2}{6k}-o(1))n}.$$
\end{theorem}

The first construction of hard instances for PPSZ was given by Chen, Scheder, Talebanfard and Tang \cite{CSTT13}. These instances are hard even for $\pstrong{\omega(1)}$.

\begin{theorem}[\cite{CSTT13}]
\label{thm:CSTT}
For any large enough $k, n > 0$ there are $k$-CNF formulas $F$ such that $$\Pr[\ppsz(F, \pstrong{n/k}) \in \sat(F)] \le 2^{-(1 - O(\log^2k/k))n}.$$ 
\end{theorem}

In this paper we improve this upper bound. For $\pweak{\omega(1)}$ we give completely different constructions for which we can show that the success probability of PPSZ is essentially tight. For $\pstrong{\omega(1)}$ we can improve the asymptotics of $k$ from $O(\log^2 k/k)$ to $O(\log k/k)$.

\begin{theorem}
\label{thm:weak}
For every $k\geq 3$ and every large enough $n$ there exists a uniquely satisfiable $k$-CNF formula $F$ on $n$ variables such that
\begin{enumerate}
\item  $\Pr[\ppsz(F, \pweak{w}) \in \sat(F)] \le 2^{-(1 - \frac{2}{k})n}$ for some $w = \Theta(\log n)$,
\item for any $\epsilon > 0$, $\Pr[\ppsz(F, \pweak{w}) \in \sat(F)] \le 2^{-(1 - \frac{2(1 + \epsilon)}{k})n}$ for some $w = n^{\Theta(\epsilon)}$ .
\end{enumerate}
\end{theorem}

In particular, for $k=3$, our upper bound is $2^{0.333\dots n}$, which is fairly close to the lower bound $2^{0.386\dots n}$ of~\cite{H14}.

\begin{theorem}
For every $k\geq 3$ and every large enough $n$, there exists a $k$-CNF
formula $F$ on $n$ variables with a unique satisfying assignment such
that $\ppsz(F,\pstrong{n/k})$ is successful with probability at most
$2^{ (-1 + \epsilon) n}$, where $\epsilon = O\pfrac{\log(k)}{k}$.
  \label{theorem-linear-strong}
\end{theorem}
\dominik{We should mention the word ``linear'' around here.}

The analysis of our hard instances is based on an encoding view of PPSZ. Given a formula $F$ on variables $x_1, \ldots, x_n$ and a satisfying assignment $\mathbf{b}$, PPSZ produces an encoding of the assignment with respect to a given permutation $\pi$ of the variables in the following way.

\begin{algorithm}[H]
$\encode (\mathbf{b}, \pi, F, P)$

$\mathbf{c} :=$ empty string\\
 \For{$i = 1, \ldots, n$}{
  \If{$P(F, x_{\pi(i)}) \not \in \{0,1\}$}{
   append $\mathbf{b}_{\pi(i)}$ to $\mathbf{c}$\;
   }
   
   $F := F|_{x_{\pi(i)} \rightarrow \mathbf{b}_{\pi(i)}}$\;
}
 
\end{algorithm}
It is not hard to see that we can express the success probability of PPSZ in terms of expected code lengths as follows.

\begin{lemma}[\cite{PPSZ05}]
Let $F$ be a $k$-CNF and let $P$ be a sound heuristic. We have $\Pr[\ppsz(F, P) \in \sat(F)] = \sum_{\b \in \sat(F)} \E_{\pi}2^{-|\encode(\b, \pi, F, P)|}$.
\end{lemma}

Thus our goal is to construct instances having a few satisfying assignments,
all of which admitting only long encodings. Defining
the optimal encoding length a satisfying assignment $\b$ to be 
$\cl(F,P,\b) := \min_{\pi} |\encode(\b,\pi,F,P)|$ we get
\begin{eqnarray*}
\Pr[\ppsz(F, P) \in \sat(F)] 
&\le&  \sum_{\b \in \sat(F)} 2^{-\cl(\b, F, P)} \ .
\end{eqnarray*}

The formulas in Theorem \ref{thm:weak} and Theorem \ref{theorem-linear-strong} have 
the unique satisfying assignment $\nullv$. Thus our goal will be to 
prove a lower bound on $\cl(F,P,\nullv)$.
%

\section{Notation and Preliminaries}

Let $F$ be a CNF formula with variable set $V$. A restriction (or partial assignment)
is a partial function $\rho: V \rightarrow \{0,1\}$. For $\b \in \{0,1\}^n$, the 
notation $S \mapsto \b$ is the restriction that maps $x \in S$ to $\b_x$
and  is undefined on $V \setminus S$. By $F|_{\rho}$ we denote the formula
arising from fixing the variables according to $\rho$ and then simplifying 
the resulting formula by removing unsatisfied literals and satisfied clauses.
For a matrix $A \in \F_2^{m \times n}$ and $U \subseteq [n]$ we denote
by $A_U$ the $(m \times |U|)$ submatrix formed by taking all columns indexed by
some $i \in U$. By $\nullv$ we denote the all-$0$-assignment as well as the
null vector in $\F_2^n$.

We will identify a vector $\a \in \F_2^{n}$ with its support $\{i \in [n] \  | \ a_i = 1\}$. 
Thus we will liberally write things like $\a \cup \b$, $\a \setminus \b$, $|\a|$, and so on.\\
\medskip

We list some key observations relating PPSZ and resolution. The (easy) proofs
can be found in the appendix.

\begin{definition}
Let $F$ be a formula with a unique satisfying assignment, which 
without loss of generality is $\nullv$, and let $P$ be a proof heuristic.
We say $F$ {\em collapses under $P$} if there is an ordering
$x_1, \dots, x_n$ of the variables in $F$ such that
$F|_{(x_1,\dots,x_{i-1} \mapsto 0)} \vdash_P (x_i = 0)$ for all $1 \leq i \leq n-1$.
\end{definition}

\begin{proposition}
 If $\cl(F, P,\b) \leq m$ then there is a set 
 $S$ of $m$ variables such that $F|_{S \mapsto \b}$ collapses under $P$.
 \label{prop-ppsz-collapse}
\end{proposition}

The next lemma states that if $F$ collapses ``sequentially'' under bounded-width
resolution, then it collapses ``simultaneously'' as well.

\begin{proposition}
Let $F$ be a $k$-CNF formula with the unique satisfying assignment $\nullv$, 
and let $w \geq k$.
  If $F$ collapses under $\pstrong{w}$ then
 $F \prstrong{w} (x = 0)$ for all variables $x$ of $F$.
 \label{prop-resolution-collapse}
\end{proposition}

The next proposition connects logical implication and collapse under $\pweak{}$ to linear algebra.
\begin{proposition}
 Let $A \in \F_2^{m \times n}$ and $F_A$ be its linear formula. 
 If $\prweak{w}(F,x_i) \in \{0,1\}$ then there is a row vector $\r \in \F_2^m$ of Hamming
 weight at most $w$ such that $\r \cdot A = \e_i$.
 \label{prop-weak-row-vector}
\end{proposition}

%



\section{Hard Instances for Weak PPSZ: Proof of Theorem \ref{thm:weak}}
\label{section-tseitin}

The  construction in this section is based on satisfiable Tseitin formulas. Unsatisfiable Tseitin formulas are extensively studied in proof complexity (see e.g. \cite{TS,U87}). 
 Given a graph $G = (V, E)$, the girth of $G$ is defined as the size of the shortest cycle in $G$. We denote this by $g(G)$.  For every pair $e, e' \in E(G)$ of edges we define the distance between $e$ and $e'$ by $\min_{u \in e, v \in e'}\{d(u, v)\}$.  We will need graphs of bounded degree with large girth. According to a well-known result of Erd\H{o}s and Sachs~\cite{ES63}, for every $k\geq 3$ and every sufficiently large $n$, there exists a $k$-regular graph with $n$ vertices and girth $>\log_{k-1}n$. 
Explicit constructions for infinitely many values of $k$ with a better constant are also known~\cite{LPS88}.


Given a degree-$k$ graph $G = (V, E)$,
 the Tseitin formula $T(G)$ is defined as follows. For each edge $e \in E$, there is a propositional variable $x_e$. For each vertex $v \in V$ we add the constraint
 $\sum_{e \ni v} x_e  = 0 \pmod 2$, which can be written as a conjunction of $2^{k-1}$ $k$-clauses.%
\footnote{The original Tseitin tautologies express the fact that the system $\sum_{e \ni v} x_e  = a_v \pmod 2$ is unsatisfiable if $\sum_va_v=1 \pmod 2$.}%
%
In our formulas we assume that the girth of the graph is at least $\log_{k-1}n$, where $n$ denotes the number of vertices. Furthermore, 
we add a clause $\neg x_e \vee \neg x_{e'}$ for each pair of edges $e, e'$ of distance at least $\frac{g(G)}{2} - 1$ (which is $\geq\frac{1}{2}\log_{k - 1}n - 1$). We call these clauses {\it bridges} and we denote the conjunction of all of them by $B$. Define $F_G := T(G) \wedge B$. Note that $F_G$ has $N = kn/2$ variables.

The following proposition follows readily (see the appendix for a proof.)

\begin{proposition}
\label{prp:tseitin-unique}
$F_G$ has the unique satisfying assignment $\nullv$.
\end{proposition}

%
%
%
%

We will consider PPSZ with $\pweak{w}$ when $w = O(\log n)$

\begin{lemma}\label{lm:acyclic}
In $F_G$ any encoding of the all-0 assignment has length at least $(1 - \frac{2}{k})N$ under $w \le \frac{1}{2}\log_{k-1}n - 1$.
\end{lemma}

\begin{proof}
Suppose for contradiction that $\cl(F,\pweak{w}) \leq \left(1 - \frac{2}{k}\right)N$ collapses.
Then there is a restriction $\rho$ that sets some $\left(1 - \frac{2}{k}\right)N$ 
variables to $0$ such that $F_G|_\rho$ collapses under $\pweak{w}$. Note that $\rho$ leaves
 at least $\frac{2}{k}N = n$ variables (i.e., edge) unset. Obviously this set of edges
 contains a cycle $C$. Let $\sigma := (E \setminus C \mapsto \nullv)$. Clearly $F_G|_{\sigma}$ also collapses.
 This contradicts the next lemma:
 \begin{lemma}
   Let $C \subseteq E$ be a cycle and $\sigma := (E \setminus C \mapsto \nullv)$. 
   Then $F_G|_{\sigma}$ does not collapse
   under $\pweak{w}$.
 \end{lemma}
\begin{proof}
 Suppose it does collapse. Then there exists an edge $e \in C$ and a set
 $F' \subseteq F|_{\sigma}$ of at most $w$ clauses such that $F' \models \neg x_e$.
%
%
The clauses in $F'$ are either coming from the Tseitin part or from the bridges. Consider a path $P = v_1, \ldots, v_s$ of maximum length on which $e$ appears and the vertices of $P$ are mentioned by Tseitin clauses in $F'$. Note that $P$ cannot contain the whole cycle, since otherwise there would be too many clauses in $F'$. Let $v_0$ and $v_{s+1}$ be vertices on $C \setminus P$ connected to $v_1$ and $v_s$, respectively. We extend $P$ by $v_0$ and $v_{s+1}$. Since $w \le \frac{1}{2}\log_{k-1}n - 1$, there is no bridge between any pair of edges appearing on $P$. We can now simply set all the variables in $P$ to 1 and all other variables to 0. This would satisfy $F'$ and yet it sets $x_e$ to 1, contradicting that $F' \models \neg x_e$.
 \end{proof}
 This concludes the proof of Lemma~\ref{lm:acyclic}.
\end{proof}

Below we show that it is possible to obtain similar lower bounds even when $w$ is some function in $n^{O(\epsilon)}$. 

\begin{lemma}
\label{lm:tree}
For every $\epsilon > 0$ and every sufficiently large $n$, any encoding of the all-0 assignment with $w < n^{\frac{\epsilon}{8(k-1)}}$ has length at least $(1 - \frac{2(1+\epsilon)}{k})N$.
\end{lemma}

\begin{proof}

Let $S$ be the set of edges appearing in any encoding of the all-0 assignment. We will show that $|E \setminus S| < (1 + \epsilon)n$. Assume for a contradiction that $|E \setminus S| \ge (1 + \epsilon)n$. We will show that $E \setminus S$ contains a large subgraph which is expanding in a certain sense.

\begin{definition}
In a graph $G$ we say that a path $P = v_1, \ldots, v_t$ is {\it slender} if for all $1 \le i \le t$ we have $d(v_i) \le 2$.
\end{definition}

\begin{lemma}\label{lm:slender}
Let $G = (V, E)$ be a graph on $n$ vertices such that $|E| \ge (1 + \epsilon)n$ for some $\epsilon > 0$. There exists an induced subgraph $H \subseteq G$ on at least $\Omega(\epsilon^{3/4}n^{1/4})$ vertices with $\delta(G) \ge 2$ with no slender path of length $\geq 2/\epsilon$.
\end{lemma}

\begin{proof}
See the appendix for a proof.
%
\end{proof}

Applying Lemma \ref{lm:slender} on $E \setminus S$ we obtain a subgraph $H$ with minimum degree at least 2 which does not contain any slender path of length $\geq 2/\epsilon$. Setting all edges outside of $H$ to 0, we obtain that there exists a set of at most $w$ clauses $F'$ in the restricted formula which implies $x_e = 0$ for some $e \in H$. Let $e = (u, v)$. We will construct a tree $T_{uv}$ in $H$ by growing two disjoint rooted trees $T_u$ and $T_v$, starting at $u$ and $v$, respectively. The crucial requirement is that in both $T_u$ and $T_v$ any path of length $\geq 2/\epsilon$ that goes downwards in the rooted tree there exists a vertex of degree $\geq 3$. We call such a vertex a {\it branching} vertex. Furthermore, in $T_{uv}$ the distance between the first branching vertices in $T_u$ and $T_v$ is at most $2/\epsilon$. Using the fact that the minimum degree in $H$ is at least 2 and it does not contain any slender path of length $2/\epsilon$ and that the girth is at least $\log_{k-1}n$, we can easily construct $T_{uv}$ so that each root to leaf path in both $T_u$ and $T_v$ has $\frac{\epsilon}{8} \log_{k - 1}n$ branching vertices. Since the horizon $w < n^{\frac{\epsilon}{8(k-1)}}$, there are vertices $u'$ and $v'$ in $T_u$ and $T_v$, respectively, that are not mentioned in $F'$. Consider the unique path between $u'$ and $v'$ in $T_{uv}$. Note that this path has length at most $\frac{1}{2}\log_{k-1}n$. However, since we put bridges only between edges of distance more that $\frac{1}{2}\log_{k-1}n$, there is no bridge between any pair of edges on this path. Setting all edges on the path including $e$ to 1 and everything else to 0 satisfies $F'$, contradicting to $F' \models \neg x_e$.
\end{proof}

Lemma \ref{lm:acyclic} implies that 
$\cl(F_G, \pweak{w}) \geq  (1 - \frac{2}{k})N$  for $w \le \frac{1}{2}\log_{k-1}n - 1$.
Similarly, Lemma \ref{lm:tree} implies that 
$\cl(F_G, \pweak{w}) \geq  (1 - \frac{2(1+\epsilon)}{k})N$ for 
$w < n^{\frac{\epsilon}{8(k-1)}}$.
This completes the proof of Theorem \ref{thm:weak}.



\section{Hard Linear Formulas for Strong PPSZ: Proof of Theorem \ref{theorem-linear-strong}}
\label{section-strong}

\dominik{Required notation: $A_U$, vectors as sets, $F|_{S \mapsto \b}$}

\dominik{Maybe we should explain this already in the introduction.}
Suppose $A \in \mathbf{F}^{m \times n}$ is a matrix in which every row
has Hamming weight at most $k$. Then the system $A \cdot \x = \nullv$
consists of $m$ linear equations over $n$ variables, each of which involves at most
$k$ variables. One can encode it as a $k$-CNF formula with
$2^{k-1} \cdot m$ clauses. Let us denote this formula by $F_A$.
A CNF formula which in this way encodes a system of linear
equations will be called a {\em linear CNF} formula.

\subsection{Robust Expanding Matrices}
\label{section-code-matrices}

As often in the realm of resolution, our proof of hardness relies on a certain
notion of expansion. Loosely speaking, a matrix $A$ is a {\em robust expander}
if for every ``sufficiently large'' submatrix $A_U$ and every ``sufficiently diverse''
set of row vectors $\u_1,\dots,\u_\ell$ at least one of the vectors 
$\u_i \cdot A_U$ has ``large'' Hamming weight. We will now define this notion formally. Throughout this section, let $k \in \N$ be arbitrary but fixed (this
is the $k$ for which we want to construct hard $k$-CNF formulas).
A sequence $\u_1,\dots,\u_\ell \in \F_2^n$ is 
{\em well-increasing} if $n/k \leq |\u_i \setminus (\u_1 \cup \dots \cup \u_{i-1})| \leq 4n/k$
for every $1 \leq i \leq \ell$. This is what we mean by ``sufficiently diverse''.
%
 
 \begin{definition}[Robust Expanders]
  A matrix $A \in \F_2^{n \times n}$ is called a $t$-robust $(\ell,w)$-expander
  if for every $U \subseteq [n]$ of size $t$ and well-increasing
  sequence $\u_1\,\dots, \u_{\ell}$, there is some $1 \leq i \leq \ell$ such that
  $|\u_i \cdot A_U | > w$. 
 \end{definition}
 
 \begin{theorem}[Robust Expanders Are Hard]
   Let $t,w \in \N$, $w \geq 2n/k$, and $\ell := \floor{\frac{k \cdot t}{4n}}$. If $A$ is a $t$-robust
   $(\ell,w)$-expander, then $\cl(F, \pstrong{w},\nullv) \geq n-t$.
   \label{theorem-expanders-are-hard}
 \end{theorem}
 
 \begin{theorem}[Robust Expanders Exist]
   For every sufficiently large $n$, there is a matrix $A \in \F_2^{n \times n}$
   such that (1) every row of $A$ has Hamming weight at most $k+1$; 
   (2) the rank of $A$ is at least $ n - 2 \log(n)$; (3) $A$ is a $t$-robust
   $(\ell,w)$-expander for $t = \frac{\ttconstant \cdot \log(k)}{k} \cdot n$,
   $\ell = \floor{\frac{k \cdot t}{4n}}$, and $w = 2n/k$.
   \label{theorem-robust-expander}
 \end{theorem}
 
With these theorems we can prove Theorem~\ref{theorem-linear-strong} for strong PPSZ.
Write $t = \frac{\ttconstant \cdot \log(k)}{k} \cdot n$ and 
let $A$ be a matrix as promised by Theorem~\ref{theorem-robust-expander}.
By Theorem~\ref{theorem-expanders-are-hard} we know that
$\cl(F_A,\pstrong{2n/k}, \nullv) \geq n - t$.
The Steinitz exchange lemma from linear algebra gives us $2 \log(n)$ 
unit row vectors that we can add to $A$ to obtain a matrix
$A' \in F_2^{(n + 2 \log(n)) \times n}$ of row rank $n$. 
This means that $F_{A'}$ has the unique satisfying assignment
$\nullv$. Each added unit row vector in $A'$ is a unit clause in $F_{A'}$. It can easily
be verified that adding a unit clause reduces $\cl$ by at most $1$. Therefore
$\cl(F_{A'}, \pweak{n/k},\nullv) \geq \cl(F_A, \pweak{n/k},\nullv) - 2\log(n) \geq n - t - 2 \log(n)$.
 This proves
 Theorem~\ref{theorem-linear-strong}.

\begin{proof}[of Theorem~\ref{theorem-expanders-are-hard}]
Let $P$ be the strong proof heuristic which performs resolution of width
up to $w$. We assume that $\cl(F_A,P,\nullv) \leq n-t$ and will derive a contradiction
to the assumption that $A$ is a robust expander. 

By Proposition~\ref{prop-ppsz-collapse} and~\ref{prop-resolution-collapse},
$\cl(F_A,P,\nullv) \leq n-t$
 means that there is a partition $[n] = U \uplus S$ with $|U| = t$
such that $F' \vdash_P (x_i = 0)$ for every $i \in U$, where
$F' := {F_A}|_{S \mapsto \nullv}$ is the formula obtained from $F$ by setting
every variable in $S$ to $0$.
For notational simplicity assume $U = \{1,\dots,t\}$.
By a connection between resolution and linear algebra which is folklore by now
(see e.g. \cite{BSI10}), the fact that $F' \vdash_P (x_i = 0)$ means the following:
\begin{proposition}[Connection Between Resolution and Linear Algebra]
For every $i \in U$   there exists
   a binary tree $T_i$ in which every node $v$ is labeled with a row vector 
   $\r_v \in \F_2^n$ 
   such that:
   \begin{enumerate}
     \item for a leaf $v$, the label $\r_v$ is a unit vector,
     \item if $v$ is an inner node and $v_0,v_1$ are its children then
     $\r_v = \r_{v_0} + \r_{v_1}$.
     \item $|\r_v \cdot A_U| \leq w$ for every node $v$ of $T_i$,
     \item $\r_{\rm root} \cdot A_U  = \e_i$.
   \end{enumerate}
   We  call $T_i$ the resolution tree of $x_i$.
   \label{prop-linear-alg-resolution}
\end{proposition}

For $ i \in \{1,\dots,t\}$ let $\r_i$  be the root labels of the tree $T_i$. Since
$\r_i \cdot A_U = \e_i$ we conclude that the vectors $\r_1,\dots,\r_t$ are linearly
independent. In particular this means that $|\r_1 \cup \dots \cup \r_t| \geq t$. Equipped
with these observations and the previous proposition, we can now construct a well-increasing
 sequence $\u_1,\dots,\u_{\ell^*}$ 
with $\ell^* := \floor{\frac{k \cdot t}{4n}}$ and $| \u_i \cdot A_U| \leq w$ for all
$1 \leq i \leq \ell^*$. This will
be a contradiction to the assumption that $A$ is a robust expander.

Start with the empty sequence
and $\ell = 0$. While $\ell < \ell^*$, we
 try to extend the current well-increasing sequence $\u_1,\dots,\u_{\ell}$  by considering two cases. 
For convenience let $\u = \u_1 \cup \dots \cup \u_\ell$. Note that
$\frac{\ell \cdot n}{k} \leq |\u| \leq \frac{4 \ell \cdot n}{k}$.

Case 1. Suppose some vector $\r_i$ among $\r_1,\dots,\r_t$ satisfies $|\r_i \setminus \u| > 2n/k$. 
Recall that $\r_i$ is the root label of the tree $T_i$. We walk from the root of $T_i$
to a leaf by always choosing the child $v$ for which the ``weight'' $|\r_v \setminus \u|$ is largest.
Note that this weight is more than $2n/k$ at the root and at most $1$ at 
a leaf. Also, in every step the weight decreases by at most a factor of $2$. Thus we 
find a node $v$ on the path for which $n/k \leq |\r_v \setminus \u| \leq 2n/k$. 
We  set $\u_{\ell + 1} := \r_v$ and 
see that the sequence $\u_1,\dots,\u_{\ell+1}$ is well-increasing by the choice of $\u_{\ell+1}$.
Also, since $\u_{\ell+1}$ is the label of a node in a resolution tree, it holds that
$| \u_{\ell+1} \cdot A_U| \leq w$.

Case 2. Suppose $| \r_i \setminus \u| \leq 2n/k$ for all $1 \leq i \leq t$. 
Since $|(\r_1\cup\dots\cup\r_t) \setminus \u| \geq t - |\u|
\geq t - \frac{4 \ell \cdot n}{k} \geq
t - \frac{4 \ell^* \cdot n}{k } + \frac{4 n }{k} \geq \frac{4n}{k}$,
we can find a subset $I \subseteq [t]$ with $|I| \leq 2n/k$
such that $2n/k \leq \left| \bigcup_{i \in I} \r_i \setminus \u\right| \leq 4n/k$.
If we let $v$ be a random linear combination of the $\r_i, i \in I$, we
see that $\E[ |\boldv \setminus \u|] \geq n/k$. Thus,
there is some vector $\boldv$ which is a linear combination of the $\r_i, i \in I$ and 
$n/k \leq |\boldv \setminus \u| \leq 4n/k$.
Furthermore, 
since $|\r_i \cdot A_U| = |\e_i| = 1$ we get 
$|\boldv \cdot A_U| \leq \sum_{i \in I} | \r_i \cdot A_U| = |I| \leq 2n/k \leq w$.  We can
extend the sequence $\u_1,\dots,\u_{\ell}$ by setting $\u_{\ell+1} = \boldv$.

To summarize, this iteratively constructs a well-increasing sequence $\u_1,\dots,\u_{\ell^*}$
with $|\u_i \cdot A_U| \leq w$. We obtain a contradiction to the 
assumption that $A$ is a robust expander, which completes the proof.
\qed
\end{proof}

\subsection{Robust Kernel Expanders Exist---Proof of 
Theorem~\ref{theorem-robust-expander}}
\begin{proof}[Proof of Theorem~\ref{theorem-robust-expander}]
  We will show that a matrix $A$ sampled from a suitable probability distribution
  is a $t$-robust $(\ell,w)$ expander with high probability,
  for $t = \frac{\ttconstant \cdot \log(k)}{k} \cdot n$, 
  $\ell = \llconstant \log(k)$, 
   and $w = 2n/k$.
  Note that by definition, this will also be a $t$-robust $(\ell',w)$-expander for 
  every $\ell' \geq \ell$, thus also for $\ell' = \floor{\frac{k \cdot t}{4n}} = 
  \floor{\frac{\ttconstant}{4} \cdot  \log(k) }\geq \llconstant \log(k) = \ell$.
   
  Take a step $k$ random walk in the Hamming cube $\{0,1\}^n$ and let $\mathbf{X}$ 
  be its endpoint. We view $\mathbf{X}$ as a row vector in $\F_2^n$.
  Repeating this experiment $n$ times independently gives $n$ row vectors
  that form a matrix $B \in \F_2^{n \times n}$. Surely each row of $B$
  has Hamming weight at most $k$, and $B$ turns out to be a robust expander.
  Unfortunately its kernel will have dimension $\Theta\pfrac{\log^2(k) n }{k}$ 
  on expectation---too large for our purposes.
  We introduce a nice trick that boosts the rank of $B$.
%
%
%
%

\begin{lemma}
\label{lemma-added-diagonal}
  Let $B \in \F_2^{n \times n}$ be a matrix and let $P$ be a random permutation matrix.
  Then $\E[|\ker(B+P)|] \leq n+1$.
\end{lemma}

The proof of this lemma can be found in the appendix. We set $A := B+P$. 
By Markov's inequality, $|\ker(A)| \leq n^2$ with high probability, and therefore
also $\rank(A) \geq n - 2 \log(n)$ with high probability. Also, each row of $A$
has Hamming weight at most $k+1$.
It remains to show that $A$ has the desired expansion properties.
 First we fix a set $U$ of size $t$ and a well-increasing sequence
 $\u_1,\dots,\u_{\ell}$ and estimate the probability that
 $|\u_i \cdot A_U| \leq w$ for all $i$. For this we 
 need the following fact about random walks in the Hamming cube.
\begin{lemma}[Hamming Cube Mixing Lemma]
     Let $U \subseteq [n]$ and $\z \in \{0,1\}^U$. Let $\x$ be the endpoint of a 
     length $d$ random walk in $\{0,1\}^n$ starting at $\nullv$. Then
     $\Pr[\x_U = \z] \leq 2 \pfrac{ 1 + \left(1 - 2/n\right)^d}{2}^{|U|}$.
     \label{lemma-cube-mixing}
\end{lemma}
In particular if $d \geq n$ and $|U| = t$ is sufficiently large,
 then this probability is at most $2^{ - 2t / 3}$.
From this lemma it is easy to show the following:
\begin{lemma}
Let $\u_1, \dots, \u_{\ell}$ be a well-increasing sequence. Then 
the probability that $|\u_i \cdot A_U| \leq w$ for all $i$ is at most 
$2^{- \frac{2\cdot \ell \cdot t}{3}} \cdot {t \choose \leq w}^{\ell}$.
\label{weak-ppsz-probability}
\end{lemma}
A full proof of this lemma is in the appendix.
To prove the theorem, it remains to do a union bound over the choices of $U \subseteq [n]$
and the well-increasing sequence. The number of ways to choose $U \subseteq [n]$ of 
size $t$ is ${n \choose t} \leq \pfrac{en}{t}^t \leq k^t 
 $.
Bounding the number of well-increasing sequences is more subtle. 

\begin{lemma}
 The number of well-increasing sequences is at most
 $k^{\frac{4 \ell n}{k}} \cdot 2^{\frac{4 \ell^2}{k} \cdot n}$.
\end{lemma}
\begin{proof}
First, write
 $\u := \u_1 \cup \dots \cup \u_{\ell}$ and note that $| \u | \leq \frac{4 \ell n}{k}$. Thus, the 
number of possible $\u$ is at most 
${n \choose \leq {\frac{4 \ell n}{k}}} \leq k^{\frac{4 \ell n}{k}}$
Once 
we have chosen $\u$, there are at most $2^{| \u|}$ choices for each 
individual $\u_i$ and at most $2^{\ell \cdot |\u|}
\leq 2^{\frac{4 \ell^2}{k} \cdot n}$
well-increasing sequences. 
\qed
\end{proof}

Let us now multiply (1) the number of choices for $U$, (2) the number of well-increasing
sequences, and (3) for a fixed $U$ and well-increasing sequence $\u_1,\dots,\u_\ell$,
the probability that $|\u_i \cdot A_U| \leq w$. We see that this is at most
\begin{align*}
 & k^t \cdot k^{\frac{4 \ell n}{k}} \cdot 2^{\frac{4 \ell^2}{k} \cdot n}
 \cdot 2^{- \frac{2\cdot \ell \cdot t}{3}} \cdot {t \choose \leq w}^{\ell} \\
& = 
 2^{ \frac{\log^2(k) \cdot n}{k} \cdot \left( \ttconstant  + 4 \cdot \llconstant
+ 4 \cdot \llconstant^2 - \frac{2}{3} \llconstant \cdot \ttconstant  + 2\cdot  \llconstant \right)}  = o(1) \ .
\end{align*}
Here we used ${t \choose \leq w} \leq \pfrac{et}{w}^w \leq k^{2n/k}$. 
We conclude that $A$ has the desired expansion properties with high probability.
In addition, it has rank at least $n - 2 \log(n)$, and every row has Hamming 
weight at most $k+1$. This concludes the proof.
\end{proof}

\newpage

\bibliographystyle{alpha}
\bibliography{ref}

\appendix

\section{Remaining Proofs from Section~\ref{section-tseitin}---Improving the Parameters for Weak PPSZ}

 \begin{oneshot}{Proposition~\ref{prp:tseitin-unique}}
 $F_G$ has the unique satisfying assignment $\nullv$.
 \end{oneshot}
\begin{proof} 
For an assignment $\alpha$ let $G_\alpha$ denote the spanning subgraph of  $G$ containing
the edges $e$ with $\alpha(e) = 1$. Note that $\alpha$ satisfies $T(G)$ if and only if
$G_\alpha$ is {\em even}, i.e., every vertex has even degree. There are two
cases: either $G_\alpha$ is the empty graph, in which case $\alpha = \nullv$ and 
satisfies $F_G$, too. Or $G_\alpha$ contains a cycle $C$, which has length at least $g(G)$
and therefore contains a bridge. In this case, $\alpha$ violates $B$.  
%
%
%
%
\end{proof}

\begin{oneshot}{Lemma~\ref{lm:slender}}
Let $G = (V, E)$ be a graph on $n$ vertices such that $|E| \ge (1 + \epsilon)n$ for some $\epsilon > 0$. There exists an induced subgraph $H \subseteq G$ on at least $\Omega(\epsilon^{3/4}n^{1/4})$ vertices with $\delta(G) \ge 2$ with no slender path of length $2/\epsilon$.
\end{oneshot}

\begin{proof}
Let  $r = 2/\epsilon$. We first find a subgraph of minimum degree at least 2 on at least $\Omega(\sqrt{n/r})$ vertices with many edges. To do this we can remove vertices of degree at most 1 at a time. Having removed $t$ vertices we are left with a graph on $n - t$ vertices and at least $(1 + \frac{2}{r})n - t$ edges. It holds that $(1 + \frac{2}{r})n - t \ge (1 + \frac{2}{r})(n - t)$. As the remaining graph has at most $n - t \choose 2$ edges we have $(1 + 2/r)n \le {n - t \choose 2} + t$. This implies $t \le n -  \Omega(\sqrt{n/r})$. Let $n' = n - t$. We thus have $n' \ge \Omega(\sqrt{n/r})$.

If the remaining graph has no slender path of length $r$ we are done. Otherwise let $v_1, \ldots, v_{t_1}$ be a maximal slender path, i.e., $d(v_i) = 2$ for all $1 \le i \le t_1$ and $v_1$ and $v_{t_1}$ have a neighbor (possibly the same) outside $P$ of degree at least 3. We remove $v_1, \ldots, v_{t_1}$ from the graph. If there are any vertices of degree 1 we remove them one at a time until there are no more such vertices. Let the total number removed vertices be $t'_1$. We repeat this for $d$ rounds until there are no more slender paths of length $r$ and all vertices have degree at least 2. Let $t_i$ and $t'_i$ be defined similarly for the $i$th iteration. We have $t'_i \ge r$ and thus $d \le n'/r$.  Note that the total number of removed edges is $t'_1 + \ldots + t'_d + d$ and hence at most $n'+ n'/r$. We are left with a graph with at least $n'/r$ edges and hence at least $\Omega(\sqrt{n'/r}) = \Omega(\epsilon^{3/4}n^{1/4})$ vertices.
\end{proof}

\section{Strong PPSZ---Remaining Proofs for Section~\ref{section-code-matrices}}

\begin{oneshot}{Lemma~\ref{lemma-added-diagonal}}
  Let $B \in \F_2^{n \times n}$ be a matrix and let $P$ be a random permutation matrix.
  Then $\E[|\ker(B+P)|] \leq n+1$.
\end{oneshot}

\begin{proof}
  The kernel of a matrix $A \in \F_2^{n \times n}$ is the set $\{\x \in \F_2^n \ | \ A \cdot \x = \nullv\}$.
  With linearity of expectation we calculate:
  \begin{align*}
  \E[|\ker(B+P)|] & = \sum_{\x \in \F_2} \Pr[ (B+P) \cdot \x = 0] \\
  & = \sum_{\x \in \F_2} \Pr[B \cdot \x = P \cdot \x]
  \end{align*}
  Note that $B \cdot \x$ is a fixed vector whereas $P \cdot \x$ is a uniformly distributed
  over all vectors of weight $|x|$. Thus, the probability that this happens to be $B \cdot \x$
  is exactly ${n \choose |\x|}^{-1}$ if $|B \cdot \x| = |\x|$ and $0$ otherwise. Thus the above is at most
  \begin{align*}
    \sum_{w = 0}^n \sum_{\x \in \F_2^n : |\x| = w} {n \choose w}^{-1} 
     = n+1 \ .
  \end{align*}
  \qed
\end{proof}

\begin{oneshot}{Lemma~\ref{weak-ppsz-probability}}
Let $\u_1, \dots, \u_{\ell}$ be a well-increasing sequence. Then 
the probability that $|\u_i \cdot A_U| \leq 2n/k$ for all $i$ is at most 
$2^{- \frac{2\cdot \ell \cdot t}{3}} \cdot {t \choose 2n/k}^{\ell}$.
\end{oneshot}

  \begin{proof}
    Let $\mathcal{E}_j$ be the event that $|\u_i \cdot A_U| \leq 2n/k$
    for all $1 \leq i \leq j$. We want to bound $\Pr[\mathcal{E}_\ell] = 
    \prod_{j=1}^\ell \Pr[|\u_j \cdot A_U| \leq 2n/k \ | \ \mathcal{E}_{j-1}]$. 
    The lemma will follow directly from this claim:
   \paragraph{Claim.}
      For each $1 \leq j \leq \ell$ the probability
      $\Pr[|\u_j \cdot A_U| \leq 2n/k \ | \ \mathcal{E}_{j-1}]$ is at most
      $2^{-2t/3} \cdot {t \choose \leq 2n/k}$.
      	 \medskip
     
   {\em Proof of the claim.} We divide $\u_i$ into an ``old part'' $v_i$ and a ``new part'' $w_i$.
   Formally, we write $\boldv_i = \u_i \cap (\u_1 \cup \dots \u_{i-1})$
   and $\w_i = \u_i \setminus (\u_1 \cup \dots \u_{i-1})$. We 
   know that $|\w_i| \geq n/k$ since the sequence is well-increasing. Also,
   $\u_i = \boldv_i + \w_i$.
   Let $\y \in \F_2^t$ be a fixed vector. Note that
   \begin{align*}
   \Pr[\u_j \cdot A_U = \y \ | \ \mathcal{E}_{j-1}]  = 
      \Pr[\w_j \cdot A_U = \boldv_j \cdot A_U  + \y \ | \ \mathcal{E}_{j-1}]
   \end{align*}
   Now $\boldv_{j} \cdot A_U$ and $\mathcal{E}_{j-1}$ both only depend on the
   rows $\a_h$ of $A$ with $h \in \boldv_j$, and
    $\w_{j} \cdot A_U$ is independent these. Thus, it suffices
   to bound $\Pr[\w_j \cdot A_U = \z]$ for some unknown but fixed vector $\z$.
   Remember that $A = B+P$ where $P$ is a random $n \times n$ permutation matrix.
   \begin{align*}
   \Pr[\w_j \cdot A_U = \z] & = \Pr[ \w_j \cdot B_U = \z + P_U \cdot \w_j]
   \end{align*}
   What is the distribution of $\w_j \cdot B_U$? It is the sum of
   $|\w_j|$ rows of $B_U$ and thus distributed like the endpoint of a 
   $|\w_j| \cdot k \geq n$ step random walk in $\{0,1\}^n$ starting at $\nullv$ and 
   then projected to the coordinates in $U$.
   By  the Hamming Cube Mixing Lemma
   (Lemma~\ref{lemma-cube-mixing})  with $d = n$ we get
   \begin{align*}
   \Pr[\w_j \cdot A_U = \z] & \leq 2 \pfrac{1 + \left(1 - \frac{2}{n}\right)^n}{2}^{t}
   \leq 2^{-2t/3} \ .
   \end{align*}
   We conclude that $\Pr[\u_j \cdot A_U = \y \ | \ \mathcal{E}_{j-1}] \leq
   2^{-2t/3}$ for every fixed $\y \in \F_2^t$. Therefore
   \begin{align*}
   \Pr[|\u_j \cdot A_U| \leq 2n/k \ | \ \mathcal{E}_{j-1}] & \leq
  2^{-2t/3} \cdot  {t \choose \leq 2n/k}  \ .
   \end{align*}
   This proves the claim. Via the chain rule, the claim
   immediately implies the lemma. \qed
  \end{proof} 

\section{Random Walks on the Hamming Cube}

\begin{lemma}[Hamming Cube Mixing Lemma]
     Let $U \subseteq [n]$ and $\z \in \{0,1\}^U$. Let $\x$ be the endpoint of a 
     length-$d$ random walk in $\{0,1\}^n$ starting at $\nullv$. Then
     \begin{align*}
     \Pr[\x_U = \z] \leq 2 \pfrac{ (1 + \left(1 - \frac{2}{n}\right)^d}{2}^{|U|} \ .
     \end{align*}
     \label{lemma-random-walk}
\end{lemma}
%
\begin{proof}
Let $Q$ be the random walk matrix of the $n$-dimensional
Hamming cube. That is, $Q_{\vec{x},\vec{y}} = 1/n$ 
if $\vec{x}$ and $\vec{y}$ have Hamming distance $1$, and $0$
otherwise. Note that $Q$ is a $(2^n \times 2^n)$-matrix, i.e.,
it takes as input vectors of dimension $2^n$, or equivalently,
functions from $\ftwo^n$ to $\mathbf{R}$.
If $f: \ftwo^n \rightarrow [0,1]$ is a 
probability distribution over $\ftwo^n$, then $Q^d f$ is the 
distribution that we get when sampling $\vec{x} \sim f$ 
and performing a random walk of length $d$. Let $f$ be the function
that is $1$ at $\vec{0}$ and $0$ elsewhere. For  $\vec{X}$ being the 
endpoint of an $d$-step random walk starting at $\vec{0}$, it holds that
$$
\Pr[\vec{X} = \vec{y}] = (Q^t f)(\vec{y}) \ .
$$

Fortunately, we can understand $Q^t f$,
since we know the eigenvalues of $Q$:
The Hamming cube is the Cayley graph of the additive group of $\ftwo^n$
with generating set $\{\vec{e}_1,\dots,\vec{e}_n\}$.
The reader who could not make sense of this last sentence 
may read the next couple of paragraphs. The reader who is familiar
with Cayley graphs and the discrete Fourier transform can skip them.
\begin{definition}
  For $S \subseteq [n]$, define $\chi_S: \ftwo^n
  \rightarrow \mathbb{R}$ by
  $$
  \chi_S (\vec{x}) := (-1)^{\sum_{i \in S} x_i} \ .
  $$
\end{definition}
One checks that the $\chi_S$ form an orthonormal basis
of the space of functions $\ftwo^n \rightarrow \mathbb{R}$
when we choose the following inner product:
$$
\scalar{f}{g} := \E_{\vec{x} \in \ftwo^n} [f(\vec{x})g(\vec{x})] \ .
$$
Each $\chi_S$ is an eigenvector of $Q$:
\begin{eqnarray*}
  (Q \cdot \chi_S)(\vec{x}) & = &
  \sum_\vec{y} Q_{\vec{x},\vec{y}} \chi_S(\vec{y})
  = \sum_{\vec{y}: d_H(\vec{x},\vec{y})=1} \frac{1}{n}
  \chi_S(\vec{y})\\
  & = & \sum_{i=1}^n \frac{1}{n} \chi_S(\vec{x}+ \vec{e}_i) 
   =  \sum_{i=1}^n \frac{1}{n} \chi_S(\x) \chi_S(\vec{e}_i) \\
  & = & \chi_S(\x) \frac{1}{n} \sum_{i=1}^{n} \chi_S(\vec{e}_i) 
\end{eqnarray*}
So $\lambda_S := \frac{1}{n} \sum_{i=1}^n \chi_S(\vec{e}_i)$
is the eigenvector of $\chi_S$. Let us evaluate $\lambda_S$:
\begin{eqnarray*}
  \lambda_S & = & \frac{1}{n} \sum_{i=1}^n \chi_S(\vec{e}_i)
   =  \frac{1}{n} \sum_{i=1}^n (-1)^{[i \in S]} 
   =  \frac{1}{n} ( n - 2|S|) = 1 - \frac{2|S|}{n} \ .
\end{eqnarray*}
Let $f: \ftwo^n \rightarrow \mathbb{R}$ 
be the function that is $1$ on $\vec{0}$ and $0$ otherwise.
To understand $Q^t f$, we write $f$ in the basis of the 
eigenvectors of $Q$. Since the $\chi_S$ are orthonormal
under the scalar product $\scalar{\cdot}{\cdot}$, we can write
$$
f = \sum_{S \subseteq [n]} \hat{f}_S \chi_S \ ,
$$
where the coefficients $\hat{f}_S$ are
$$
\hat{f}_S := \scalar{f}{\chi_S} = 
\E_{\vec{x} \in \ftwo^n}[ f(\vec{x}) \chi_S(\vec{x})] 
=  2^{-n} ,
$$
since $\vec{x}=\vec{0}$ is the only element that contributes to
the expectation. Thus,
\begin{eqnarray*}
Q^t f & = & Q^t \left(\sum_S \hat{f}_S \chi_S\right) 
= Q^t \left(\sum_S 2^{-n} \chi_S \right) 
 = 2^{-n} \sum_S \lambda_S^t \chi_S  \ .
\end{eqnarray*}

For $\vec{y} \in \ftwo^n$, let us bound the probability 
$\Pr[\vec{X} = \vec{y}]$: With the above equation, we get
\begin{eqnarray*}
(Q^t f)(\vec{y})  & = & 2^{-n} \sum_{S} \lambda_S^t \chi_S(\vec{y}) 
 =  2^{-n} \sum_{S} \left(1 - \frac{2|S|}{n}\right)^t \chi_S(\vec{y})\\
& \leq & 2^{-n} \sum_{s=0}^n {n \choose s} 
\left|1 - \frac{2s}{n}\right|^t  \quad
\textnormal{(since } |\chi_S(\vec{y})| = 1 \textnormal{)} \\
& \leq & 2^{-n} \sum_{s=0}^{\floor{n/2}} {n \choose s} 
\left|1 - \frac{2s}{n}\right|^t 
+ 2^{-n} \sum_{s=\ceil{n/2}}^n {n \choose s} 
\left|1 - \frac{2s}{n}\right|^t \\
& = & 2^{-n} \sum_{s=0}^{\floor{n/2}} {n \choose s} 
\left|1 - \frac{2s}{n}\right|^t 
+ 2^{-n} \sum_{r=0}^{\floor{n/2}} {n \choose n-r} 
\left|1 - \frac{2(n-r)}{n}\right|^t \\
 & =  & 2 \cdot 2^{-n} \sum_{s=0}^{\floor{n/2}} {n \choose s} 
\left(1 - \frac{2s}{n}\right)^t \\
& \leq & 2 \cdot 2^{-n} \sum_{s=0}^{\floor{n/2}} {n \choose s}
\left(1 - \frac{2}{n}\right)^{st}
 \leq  2 \cdot 2^{-n} \sum_{s=0}^{n} {n \choose s}
\left(1 - \frac{2}{n}\right)^{st} \\
& = &  2 \pfrac{1 + \left(1 - \frac{2}{n}\right)^{t}}{2}^n \ .
\label{prob-at-0}
\end{eqnarray*}

This proves the lemma for $U = [n]$. In general, however,
we are interested in the distribution of $\mathbf{X}_U$, i.e.,
$\mathbf{X}$ projected to the coordinates in $U$.

\begin{observation}
  Perform a ``lazy'' random walk on $\{0,1\}^{|U|}$ as follows:
  Start at $\vec{0}$. At each step, take each edge with
  probability $1/n$. With the remaining probability $1-|U|/n$,
  don't move in this step. 
  Then the end point of this walk after $t$ steps 
  has distribution $\vec{X}_{U}$.
\end{observation}

Let $Q$ be transition matrix of 
the random walk on $\{0,1\}^{|U|}$. Then 
$$
\tilde{Q} := \frac{|U|}{n}Q + \frac{n- |U|}{n} I
$$
is the transition matrix of the lazy random walk described above.
For each $S \subseteq U$, $\chi_S$ is an eigenvector of $Q$, 
and the corresponding eigenvalue is 
$\lambda_S = 1 - \frac{2|S|}{|U|}$. The matrix $\tilde{Q}$ has
the same eigenvectors as $Q$, and its eigenvalues are
$$
\tilde{\lambda}_S = \frac{|U|}{n} \lambda_S + \frac{n-|U|}{n} \cdot 1
= 1 - \frac{2|S|}{n} \ .
$$
Let $f: \{0,1\}^{|U|}  \rightarrow \mathbb{R}$ be the function
that is $1$ at $\vec{0}$ and $0$ elsewhere. We write
$f = \sum_{S \subseteq U} \hat{f}_S \chi_S$. Let $u := |U|$. 
By the same calculation as above, $\hat{f}_S = 2^{-u}$. Thus,
for $\vec{y} \in \{0,1\}^{u}$ we get
\begin{eqnarray*}
(\tilde{Q}^t f)(\vec{y})  & = & 2^{-u} \sum_{S} \lambda_S^t \chi_S(\vec{y}) 
 =  2^{-u} \sum_{S} \left(1 - \frac{2|S|}{n}\right)^t \chi_S(\vec{y})\\
& \leq & 2^{-u} \sum_{s=0}^u {u \choose s} 
\left|1 - \frac{2s}{n}\right|^t  \quad
\textnormal{(siuce } |\chi_S(\vec{y})| = 1 \textnormal{)} 
\end{eqnarray*}
If $u \leq n/2$, we observe that all eigenvalues $1 - 2s/n$ are
non-negative.\footnote{The reader might observe that in our application
indeed $|U| \ll n/2$.} In this case we continue:
\begin{eqnarray}
2^{-u} \sum_{s=0}^u {u \choose s} 
\left(1 - \frac{2s}{n}\right)^t  
\leq
2^{-u} \sum_{s=0}^u {u \choose s} 
\left(1 - \frac{2}{n}\right)^{st}
= \pfrac{1 + \left(1 - \frac{2}{n}\right)^t}{2}^{u}
\label{ineq-small-u}
\end{eqnarray}
and we are done. If $u > n/2$, things get more tricky. We split the 
sum in two parts:
\begin{eqnarray}
2^{-u} \sum_{s=0}^{\floor{n/2}} {u \choose s} 
\left(1 - \frac{2s}{n}\right)^t 
+ 2^{-u} \sum_{s=\floor{n/2}+1}^u {u \choose s} 
\left(\frac{2s}{n}-1\right)^t 
\label{two-sums}
\end{eqnarray}

We can bound the first sum exactly similar as in (\ref{ineq-small-u}):
\begin{align*}
2^{-u} \sum_{s=0}^{\floor{n/2}} {u \choose s} 
\left(1 - \frac{2s}{n}\right)^t & \leq
2^{-u} \sum_{s=0}^{\floor{n/2}} {u \choose s} 
\left(1 - \frac{2}{n}\right)^{st} \leq
2^{-u} \sum_{s=0}^{u} {u \choose s} 
\left(1 - \frac{2}{n}\right)^{st} \\
& =  \pfrac{1 + \left(1 - \frac{2}{n}\right)^t}{2}^{u} \ .
\end{align*}
 Let us bound the second sum in (\ref{two-sums}). For
notational convenience, we let it run from $\ceil{n/2}$ to $u$,
only making it larger. We change the parameter $s$ to 
$r := u-s$. Thus 
\begin{align*}
2^{-u} \sum_{s=\ceil{n/2}}^u {u \choose s} 
\left(\frac{2s}{n}-1\right)^t
& = 
2^{-u} \sum_{r=0}^{u-\ceil{n/2}} {u \choose {u-r}} 
\left(\frac{2(u-r)}{n}-1\right)^t\\
& = 
2^{-u} \sum_{r=0}^{u-\ceil{n/2}} {u \choose {r}} 
\left(\frac{2u-n}{n}-\frac{2r}{n}\right)^t\\
& \leq 
2^{-u} \sum_{r=0}^{u-\ceil{n/2}} {u \choose {r}} 
\left(1-\frac{2r}{n}\right)^t \tag{since $u \leq n$} \\
& \leq 
2^{-u} \sum_{r=0}^{u-\ceil{n/2}} {u \choose {r}} 
\left(1-\frac{2}{n}\right)^{rt} \\
& \leq 
2^{-u} \sum_{r=0}^{u} {u \choose {r}} 
\left(1-\frac{2}{n}\right)^{rt}  =  \pfrac{1 + \left(1 - \frac{2}{n}\right)^t}{2}^{u} \  . 
\end{align*}
Thus, both sums in (\ref{two-sums}) are bounded by (\ref{ineq-small-u}) and the lemma follows.
\end{proof}

\end{document}
\grid
\grid